\documentclass{amsart}

\usepackage{amssymb,stmaryrd,mathrsfs,dsfont}

\theoremstyle{plain}
\newtheorem{theorem}{Theorem}[section]

\theoremstyle{definition}
\newtheorem{notation}[theorem]{Notation}
\newtheorem{example}[theorem]{Example}

\theoremstyle{remark}
\newtheorem{remark}[theorem]{Remark}

\input xy
\xyoption{all}

\def\au{\mathcal{A}}

\begin{document}
\title[On the Synchronization of CSFA]{On the Synchronization of Circular Semi-Flower Automata}

\author[Shubh N. Singh]{Shubh N. Singh}
\address{Department of Mathematics, Central University of South Bihar, Gaya, India}
\email{shubh@cub.ac.in}
\author[Ankit Raj]{Ankit Raj}
\address{Department of Mathematics, Central University of South Bihar, Gaya, India}
\email{ankitraj@cusb.ac.in}


\begin{abstract}
Pin proved that every circular automaton with a prime number of states containing a non-permutation is synchronizing. In this paper, we investigate the synchronization of circular semi-flower automata. We first prove that every semi-flower automaton is a one-cluster automaton. Subsequently, we prove that every semi-flower automaton containing a $1$-cycle is synchronizing. Further, we prove that every circular semi-flower automaton with an odd number of states containing a $2$-cycle is synchronizing.
\end{abstract}

\subjclass[2010]{20M35, 68Q45, 68Q70, 68R10}

\keywords{circular automaton, semi-flower automaton, one-cluster automaton, synchronizing automaton, \v{C}ern\'{y} conjecture}

\maketitle

\section*{Introduction}

An automaton is called \emph{synchronizing} if there exists a word, called a \emph{synchronizing word}, that sends all its states to a single state. The concept of synchronization of automata has many practical applications in such areas as robotics, manufacturing, coding theory, bio-computing, model-based testing and many others \cite{volk08}.

A lot of investigations have already been done in the area of synchronization of automata. Pin \cite{pin78} proved that every circular automaton with a prime number of states containing a non-permutation letter is synchronizing. Perles et al. \cite{rabin63} observed that the class of definite automata is a sub-class of synchronizing automata. It is verified that every strongly connected aperiodic automaton is synchronizing. For a given automaton, it is often very challenging to prove that the automaton is synchronizing or not.

Besides of the applications of synchronizing automata, there is a famous conjecture, known as \v{C}ern\'{y} conjecture, related to synchronizing automata.
\v{C}ern\'{y} conjecture states that every $n$-state synchronizing automaton has a synchronizing word of length at most $(n-1)^2$ \cite{cerny64}.
In this connection, Pin \cite{pin78} proved that the \v{C}ern\'{y} conjecture is satisfied for circular automata with a prime number of states. Further, Dubuc \cite{dubuc98} showed that all circular automata satisfies the \v{C}ern\'{y} conjecture. Steinberg \cite{stein11} proved the \v{C}ern\'{y} conjecture for one-cluster automata with prime length cycle. For another special classes of synchronizing automata, the \v{C}ern\'{y} conjecture has also been verified, or sharper bounds than the general bounds have been proven, see for instance \cite{traht07, epps90, volk04, pin83, kari03, perin11, grech13, volk09}.

This paper investigates the synchronization of circular semi-flower automata. Circular automata have been studied in various contexts \cite{pin78, dubuc98}. Semi-flower automata have been introduced to study the finitely generated submonoids of a free monoid \cite{giam07,singh12}. Using semi-flower automata, the rank and intersection problem of certain finitely generated submonoids of a free monoid have been investigated \cite{giam08, singh12, shubh12}. Semi-flower automata have also been studied in many different contexts, see for instance \cite{prib11, giam08, shubh13, shubh16, shubh18}.

The remaining part of the paper is organized as follows. In Section 1, we introduce the notation and briefly give the required background. In Section 2, we investigate the synchronization of circular semi-flower automata. In Section 3, we conclude the paper and provide some future directions for our work.

\section{Preliminaries and Notation}

In this section we first briefly introduce the notations used in this paper. Throughout this paper, $n$ is an integer greater than $1$. Let $P$ be a non-empty finite set. The number of elements in $P$ is denoted by $|P|$. We write argument of a transformation $\alpha$ of $P$ on its left so that $p\alpha$ is the value of $\alpha$ at the argument $p\in P$. The composition of transformations is designated by concatenation, with the leftmost transformation understood to apply first, so that $p(\alpha \beta) = (p\alpha)\beta$. We denote by $T_n$ the full transformation monoid of a set with $n$ elements.

Let $D$ be a (labeled) digraph. The vertex set of $D$ is denoted by $V(D)$. A \emph{path} in $D$ is an alternating finite sequence $v_0, e_1, v_1, \ldots v_{k-1}, e_k, v_k$ of distinct vertices and (labeled) edges such that, for $1\le i \le k$, the tail and the head of edge $e_i$ are vertices $v_{i-1}$ and $v_{i}$, respectively. A path with at least one edge is called a \emph{cycle} if its initial vertex and terminal vertex are the same. The \emph{length} of a path is the number of its edges. A \emph{$k$-cycle} is a cycle of length $k$. If there is a path from vertex $u$ to vertex $v$, then the \emph{distance} from $u$ to $v$ is the length of shortest path from $u$ to $v$.

An \emph{alphabet} is a non-empty finite set. The elements of an alphabet are referred to as \emph{letters}. We denote an alphabet by the symbol $A$.
We use symbols $A^*$ and $\varepsilon$ to denote the set of all words over $A$ and the empty word, respectively. We consider an \emph{automaton} (over $A$) as a quintuple $\au = (Q, A, \delta, q_0, F)$, where $Q$ is the non-empty finite set of \emph{states}, $q_0 \in Q$ is the \emph{initial} state, $F \subseteq Q$ is the set of \emph{final} states, and $\delta: Q\times A \rightarrow Q$ is the \emph{transition} (total) function. An automaton with $n$ states is called an $n$-state automaton. The canonical extension of $\delta$ to the set $Q\times A^*$ is still denoted by $\delta$.

Let $\au$ be an $n$-state automaton. Each word $x \in A^*$ has a natural interpretation as a transformation of $T_n$ and we do not distinguish between the word $x$ and its interpretation. A letter $a \in A$ is called \emph{permutation} if its interpretation is a permutation; otherwise, it will be called \emph{non-permutation}. $\au$ is called \emph{circular} if there exists a permutation letter which induces a circular permutation on its set of states. The set $M(\au) = \{ x\in T_n  \;|\;x \in A^*\}$ forms a submonoid of the full transformation monoid $T_n$ called the \emph{transition monoid} of $\au$.
The automaton $\au$ is called \emph{synchronizing} if there exists a word $x\in A^*$ such that the image of $Q$ under the transformation $x \in M(\au)$ is a singleton.

Let $\au$ be an automaton. By denoting states as vertices and transitions as (labeled directed) edges, $\au$ can be represented by (labeled) digraph, denoted by $D(\au)$, in which initial state and final states shall be distinguished appropriately. For $a\in A$, we define an $a$-\emph{edge} as an edge labeled by the letter $a$. A \emph{path} (respectively, \emph{cycle}) in $\au$ is a path (respectively, a cycle) in $D(\au)$. A state $q$ is said to be \emph{accessible} (respectively, \emph{co-accessible}) if there exists a path from $q_0$ to $q$ (respectively, a path from $q$ to some final state). $\au$ is called \emph{semi-flower automaton} (in short, SFA) if $F = \{q_0\}$, every state is both accessible and co-accessible, and all cycles in $\au$ visit the unique initial-final state $q_0$. For further basic definitions concerning digraphs and automata we refer \cite{gutin02,lawson04}.

Let $G$ be a finite group with identity element $e$. We write $|G|$ to denote the order of $G$. Let $g$ be an element of $G$. The \emph{order} of the element $g$ is the smallest positive integer $t$ such that $g^t = e$. The cyclic subgroup generated by the element $g$ is denoted by  $\langle g\rangle$. The group $G$ is called  \emph{cyclic} if $G = \langle g \rangle$ for some $g\in G$. In this case, $g$ is called a \emph{generator} of $G$. It is well known that any two finite cyclic groups of the same order are isomorphic. In a finite group, the order of a group element divides the order of its group. Therefore, if $|G|$ is odd, the order of each element of $G$ is also odd. All further unexplained notation and terminology of groups we refer \cite{foote04}.

\section{Main Results}

In this section we investigate the synchronization of circular semi-flower automata (in short, CSFA). In order to investigate the synchronization of CSFA, we recall the concept of one-cluster automata introduced by B{\'{e}}al and Perrin in \cite{beal09}.

Let $\au$ be an automaton and let $\mathcal{R}$ be the sub-digraph of $D(\au)$ made of the $b$-edges. A connected component of $\mathcal{R}$ is called a $b$-\emph{cluster}. Note that each $b$-cluster contains a unique cycle, called a $b$-cycle, with possible trees attached to $b$-cycle at their root. A \emph{one-cluster automaton} with respect to a letter $b$ is an automaton which has exactly one $b$-cluster.

The following theorem proves that every SFA is a one-cluster automaton with respect to each letter.
\begin{theorem}\label{uni-cycle}
Every SFA is a one-cluster automaton with respect to each letter.
\end{theorem}
\begin{proof}
Let $\au$ be an SFA and let $b\in A$. Consider the sub-digraph $\mathcal{R}$ of the digraph $D(\au)$ made of the $b$-edges. It is sufficient to prove that the underlying graph of $\mathcal{R}$ is connected. Let $p$ be a state such that $p \neq q_0$. We claim that $p \;b^r = q_0$ for some positive integer $r$. Clearly, the state $p$ belongs to exactly one $b$-cluster. Recall that each $b$-cluster contains a unique cycle. Since $\au$ is an SFA, the initial-final state $q_0$ belongs to each cycle, and subsequently the state $q_0$ belongs to each $b$-cluster. Therefore $p \; b^r = q_0$ for some positive integer $r$, and consequently the underlying graph of $\mathcal{R}$ is connected. This proves our theorem.
\end{proof}

In view of Theorem \ref{uni-cycle}, we get that every SFA contains a unique $b$-cycle, where $b\in A$.
\begin{notation}
We shall denote by $C$ the $b$-cycle in an SFA. Further, the set of states in the $b$-cycle $C$ shall be denoted by $V(C)$.
\end{notation}
Obviously $q_0 \in V(C)$. We now recall the definition of the level of automata from \cite{beal09}. Let $\au$ be an automaton and let $b\in A$. The \emph{level} of a state $q$ in a $b$-cluster is defined as the distance between $q$ and the root of the tree containing $q$. If $q$ belongs to the cycle, its level is defined as zero. The \emph{level} of $\au$ is the maximal level of its states.

\begin{notation}
We shall denote the level of an automaton by symbol $l$.
\end{notation}
Note that $Qb^l = V(C)$. If the length of the cycle $C$ is $1$, the following theorem asserts about the synchronization of an SFA.

\begin{theorem}\label{cycle1}
Let $\au$ be an SFA. If $C$ is a $1$-cycle in $\au$, then $\au$ is synchronizing.
\end{theorem}

\begin{proof}
Recall that $q_0 \in V(C)$. Since $|V(C)| = 1$, it follows that $V(C)=\{q_0\}$. Hence, $Q b^l = V(C) = \{q_0\}$. This completes our proof.
\end{proof}

We now recall the necessary result from \cite{shubh16}.

\begin{theorem}[\cite{shubh16}]\label{c3.l.ucp}
Let $\au$ be an SFA.
\begin{enumerate}
	\item[\rm (i)] For $a \in A$, if $a$ is a permutation, then $a$ is a circular permutation.
	\item[\rm(ii)] For $a, b \in A$, if $a$ and $b$ are permutations, then the permutations $a$ and $b$ are same.
\end{enumerate}
\end{theorem}

Unless otherwise stated, in what follows, $\au$ denotes an $n$-state CSFA. In view of Theorem \ref{c3.l.ucp}, there is a unique circular permutation induced by letters. For the rest of paper, we fix the following regarding $\au$. Assume that the letter $a\in A$ induces the circular permutation, and accordingly $q_0, q_1, \ldots, q_{n-1}$ is the cyclic ordering of $Q$ with respect to $a$. We use symbol $G$ to denote the cyclic subgroup of the transition monoid $M(\au)$ generated by the permutation $a$.

For completeness, we state the following simple result from \cite{shubh18}.
\begin{remark}[\cite{shubh18}]\label{t.csfa}
Let $\au$ be an $n$-state circular semi-flower automaton. Then
\begin{enumerate}
	\item[\rm (i)] $|G| = n$.
	\item[\rm(ii)] $G$ is the group of units of $M(\au)$.
\end{enumerate}
\end{remark}

For an odd integer $n$, If the length of the cycle $C$ is $2$, the following theorem proves that an $n$-state circular semi-flower automaton $\au$ is synchronizing.

\begin{theorem} \label{main1}
For an odd integer $n$, let $\au$ be an $n$-state circular semi-flower automaton. If $C$ is a $2$-cycle in $\au$, then $\au$ is synchronizing.
\end{theorem}

\begin{proof}
Recall that $q_0 \in V(C)$. Since $|V(C)| = 2$, let $q_m\;(1\le m \le n-1)$ be another state such that $q_m \in V(C)$, and so $V(C) = \{q_0, q_m\}$. Then $Q\cdot b^l = V(C) = \{q_0, q_m\}$. Clearly, $q_m\cdot b = q_0$ and $q_0\cdot b = q_m$. It follows that
\[q_0 \cdot b^2  = q_0 \Longrightarrow q_0 \cdot b^{2l}  = q_0 \Longrightarrow q_0 \cdot b^{1+2l} = q_m \cdot b^{2l} = q_m ,\] and
\[q_m \cdot  b^2 = q_m \Longrightarrow q_m \cdot  b^{2l} = q_m \Longrightarrow q_m \cdot b^{1+2l} = q_0 \cdot b^{2l} = q_0.\]
Thus, $b^{1+2l}$ maps $Q$ into $\{q_0,q_m\}$ and swaps $q_0$ and $q_m$.

Note that $q_m \cdot a^{(n-m)} = q_0$. We now consider the sequence $\langle w_k\rangle$ of words, where $w_k := a^{k(n-m)}b^{1+2l}$. From the suffix word $b^{1+2l}$, it is clear that $\{q_0, q_m\}\cdot w_k \subseteq \{q_0,q_m\}$ for each $k$. If there exists a $k$ such that $\{q_0, q_m\}\cdot w_k$ is a singleton, then the automaton $\au$ is synchronizing.

Otherwise, for $k=1$, we have \[q_m\cdot a^{(n-m)}b^{1+2l} = q_0\cdot b^{1+2l} = q_m, \mbox{ and therefore } q_0\cdot a^{(n-m)}b^{1+2l} = q_0.\] By induction, we get that
\begin{equation}\label{odd}
q_0 \cdot a^{k(n-m)}b^{1+2l} =
\begin{cases}
q_0 & \text{if $k$ is odd}\\
q_m & \text{if $k$ is  even}.
\end{cases}
\end{equation}
Since by assumption $n$ is an odd number and the letter $a$ induces a circular permutation, the order of $a^{(n-m)}$ is also an odd number, say $t$. Now for this number $t$, we have that $a^{t(n-m)}$ induces the identity transformation. Hence \[q_0 \cdot a^{t(n-m)}b^{1+2l} = q_0 \cdot b^{1+2l} = q_m,\] a contradicting to the previous statement. This completes our proof.
\end{proof}

The following example shows that Theorem \ref{main1} is not necessarily true when circular semi-flower automaton has an even number of states.
\begin{example}\label{even}
Consider $6$-state circular semi-flower automaton $\au_1$ over $A = \{a, b\}$  given in the Figure \ref{evensic}. Clearly, $V(C) = \{q_0, q_3\}$. By using the computer algebra system  GAP--Groups, Algorithms and Programming \cite{gap} we observe that the transition monoid $M(\au_1)$ does not contain a constant transformation. Hence, $\au_1$ is non-synchronizing.
\begin{figure}[ht]
\entrymodifiers={++[o][F-]} \SelectTips{cm}{}
\[\xymatrix{*\txt{} & q_1 \ar[r]^a \ar[drr]^b & q_2 \ar[dr]^{a, \ b}& *\txt{}\\
 *++[o][F=]{q_0} \ar[ur]^a \ar@/^0.9pc/[rrr]^b & *\txt{} & *\txt{} & q_3\ar[ld]^a\ar@/^0.9pc/[lll]^b\\
 *\txt{} & q_5 \ar[ul]^{a, \ b} & q_4\ar[l]^a \ar[ull]^b & *\txt{}}\]
\caption{$6$-state non-synchronizing CSFA $\au_1$ with $|V(C)| = 2$}
\label{evensic}
\end{figure}
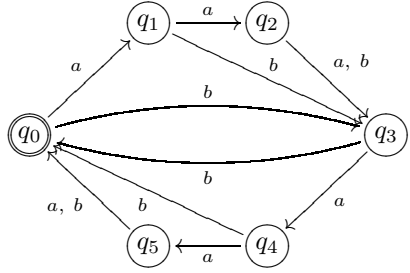
\end{example}

For an odd integer $n$, let $\au$ be an $n$-state circular semi-flower automaton. If $C$ is a $3$-cycle in $\au$, then $\au$ is not necessary synchronizing as shown in the following example.
\begin{example}
Consider $9$-state circular semi-flower automaton $\au_2$ given in the Figure \ref{nsic01}. Clearly, $V(C) = \{q_0, q_3, q_6\}$. By using the computer algebra system  GAP--Groups, Algorithms and Programming \cite{gap} we observe that the transition monoid $M(\au_2)$ does not contain a constant transformation. Hence, $\au_2$ is non-synchronizing.
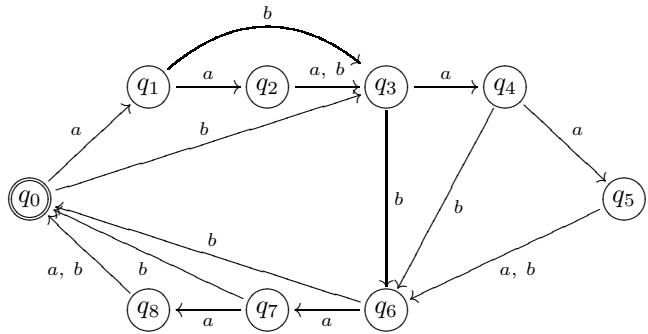
\begin{figure}[ht]
\entrymodifiers={++[o][F-]} \SelectTips{cm}{}
\[\xymatrix{*\txt{} & q_1 \ar[r]^a \ar@/^1.9pc/[rr]^b & q_2\ar[r]^{a, \ b} & q_3\ar[r]^a \ar[dd]^b& q_4\ar[dr]^a\ar[ddl]^b\\
*++[o][F=]{q_0} \ar[ur]^a \ar[urrr]^b & *\txt{} &*\txt{} &*\txt{} &*\txt{} & q_5 \ar[dll]^{a, \ b}\\
*\txt{} & q_8\ar[ul]^{a,\ b} & q_7 \ar[l]^a \ar[ull]^b & q_6\ar[l]^a \ar[ulll]_b & *\txt{} & *\txt{}}\]
\caption{$9$-state non-synchronizing CSFA $\au_2$ with $|V(C)| = 3$}
\label{nsic01}
\end{figure}
\end{example}

\section*{Conclusion and further directions}

This work investigated the synchronization of circular semi-flower automata (CSFA). We proved that every semi-flower automaton is one-cluster automaton, and subsequently we observed that every semi-flower automaton containing a $1$-cycle is synchronizing. Further, for an odd number $n$, we proved that every $n$-state CSFA containing a $2$-cycle is synchronizing. Using the computer algebra system GAP, we finally provided examples of $6$-state non-synchronizing CSFA and $9$-state non-synchronizing CSFA containing, respectively, $2$-cycle and $3$-cycle.

In the present work, we observed that an $n$-state CSFA is in general not synchronizing. Therefore, one could investigate on the sufficient conditions for the synchronization of an $n$-state CSFA.

\end{document}